\documentclass[twocolumn]{autart}    

\usepackage[final]{changes}

\definechangesauthor[color=red]{MS}
\definechangesauthor[color=blue]{AJ}

\usepackage{todonotes}

\usepackage{amsmath,array,comment}
\usepackage{times,enumerate} 
\usepackage{color}
\usepackage{graphicx}
\usepackage{setspace}
\usepackage{bbm}
\usepackage{mathdots,mathrsfs}
\usepackage{amssymb,latexsym,amsfonts,amsmath,cite,xfrac}
\usepackage{stmaryrd}
\usepackage{caption}
\usepackage[dvips]{psfrag}
\usepackage{natbib}        
\usepackage{setspace}
\usepackage{amsmath}
\usepackage{url}
\usepackage{pgf,tikz}
\usepackage[normalem]{ulem}
\usepackage{algorithm}
\usepackage{algpseudocode,graphicx}
\usepackage{subcaption}
\usepackage{multirow}
\usepackage{rotating}

\usepackage{amssymb,latexsym,amsfonts,amsmath,cite,comment}

\usepackage{ntheorem,lipsum}
\theorembodyfont{\upshape}
\newtheorem{theorem}{Theorem}
\newtheorem{definition}{Definition}
\newtheorem{assumption}{Assumption}
\newtheorem{proof}{Proof}
\newtheorem{example}{Example}
\newtheorem{remark}{Remark}

 \newcommand{\R}{{\mathbb{R}}}

\newcommand{\Z}{{\mathbb Z}}

\newcommand{\rhoo}{{\rho}}

\newcommand{\ie}{{\it i.e.}}

\usepackage{pgf,tikz}

\DeclareMathOperator{\card}{card}
\DeclareMathOperator{\Rank}{rank}

\DeclareMathOperator{\diag}{diag}

\DeclareMathOperator{\GenDualSet}{GenDualSet}

\begin{document}

\begin{frontmatter}
                                              \title{ A Separation Theorem for Joint Sensor and Actuator Scheduling\\ with Guaranteed Performance Bounds}

\thanks[footnoteinfo]{This paper was not presented at any IFAC meeting.}
\thanks{This research was supported in part by a Vannevar Bush Fellowship from the Office of Secretary of Defense.}

\vspace{-0.5cm} 

\author[NEU]{Milad Siami}\ead{m.siami@northeastern.edu},    
\author[MIT]{Ali Jadbabaie}\ead{jadbabai@mit.edu}              

\address[NEU]{Electrical \& Computer Engineering Department,\\ Northeastern University, Boston, MA 02115 USA}
\address[MIT]{Institute for Data, Systems, and Society,\\ Massachusetts Institute of Technology, Cambridge, MA 02139 USA } %

\begin{keyword}                           
Network Analysis and Control; Fundamental Limits; Sparse Sensor and Actuator Selections; The Separation Principle.               
\end{keyword}                             

\begin{abstract}  We study the problem of jointly designing a sparse sensor and actuator schedule for linear dynamical systems while guaranteeing a control/estimation performance that approximates the fully sensed/actuated setting. We further prove a separation principle, showing that the problem can be decomposed into finding sensor and actuator schedules separately. However, {it is shown that this problem cannot be efficiently solved or approximated in polynomial, or even quasi-polynomial time for time-invariant sensor/actuator schedules;} instead, we develop  deterministic polynomial-time algorithms for a time-varying sensor/actuator schedule with guaranteed approximation bounds. Our main result is to provide a polynomial-time joint actuator and sensor schedule that on average selects only a constant number of sensors and actuators at each time step, irrespective of the dimension of the system. The key idea is to sparsify the controllability and observability Gramians while providing approximation guarantees for Hankel singular values. This idea is inspired by  recent results in theoretical computer science literature on sparsification. 
\end{abstract}

\end{frontmatter}
\section{Introduction}
\allowdisplaybreaks

One of the main challenges in realizing the promise of smart urban mobility is localization, perception, mapping and control with a myriad of sensors and actuators, e.g., camera sensors, data from 3-D mapping, LIDAR, electric motor, valve, etc. A key obstacle to this vision is the information overload, and the computational complexity of perception, mapping, and control using a large set of sensing and actuating modalities. A possible solution is to find a sparse yet important subset of sensors (actuators) and use those instead of using all available measurements (actuators) \citep{matni2016regularization, argha2017framework,fahroo2000optimal}. When the dimension of the state is large, finding the optimal yet low cardinality subset of features is like finding a needle in a haystack: the problem is computationally difficult and provably NP-Hard \citep{Alex2014,vassili1}. 

Often times, we are interested in reducing the control complexity, operation cost, and maintenance cost by not using all available actuators and sensors. The choice of sensors and actuators affect the performance, computational cost, and costs of the control system. 
As it is shown recently in~\cite{Alex2014} and \cite{vassili1}, the problem of finding a sparse set of input variables such that the resulting system is controllable, is NP-hard. Even the presumably easier problem of approximating the minimum number better than a constant multiplicative factor of $ \log n$ is also NP-hard \citep{Alex2014}. Other results in the  literature have shown network controllability by exploring approximation algorithms for the closely related subset selection problem \citep{Alex2014,summers2016submodularity, pequito2015complexity,nozari2019heterogeneity, bopardikar2017sensor}. More recently, some of the authors showed that even the problem of finding a sparse set of actuators to guarantee reachability of a particular state is hard and even hard to approximate \citep{jadbabaie2017minimal}.

Over the past few years, controllability and observability properties of complex dynamical networks have been subjects of intense study in the controls community~\citep{Alex2014,pasqualetti2014controllability, liu2016control, 7962975,muller1972analysis, tzoumas2016minimal,7171062, segio2017,NozariACC,yaziciouglu2016graph,summers2016submodularity, pequito2015complexity}. This interest stems from the need to steer or observe the state of  large-scale, networked systems such as power grids \citep{chakrabortty2011control}, social networks, biological and genetic regulatory networks \citep{ChandraBD11,marucci2009turn,rajapakse2012can}, and traffic networks \citep{SiamiTAC18}.
Previous studies have been mainly focused on solving the optimal sensor/actuator selection problem using the greedy heuristic, as approximations of the corresponding sparse-subset selection problem. However, in \cite{SiamiArXiv18}, we develop a framework to design a sparse actuator schedule for a given large-scale linear system with guaranteed performance bounds using deterministic polynomial-time and randomized approximately linear-time algorithms, and we gain new fundamental insights into approximating various performance metrics compared to the case when all actuators are chosen. In \cite{tzoumas2018sensing}, the authors show that a separation principle holds for the Linear-Quadratic-Gaussian (LQG) control problem.

In this paper, we build upon our previous work \citep{SiamiArXiv18} and consider the problem of jointly designing the sparse sensor and actuator schedule for linear dynamical systems, to ensure desired performance and sparsity levels of active sensors and actuators in time and space.
The joint sensor and actuator (S/A) scheduling problem involves selecting an appropriate number, activation time, position, and type of sensors and actuators. The idea is to essentially sparsify the choice of sensor and actuators both spatially and temporally.
 We show that by carefully designing a time-varying joint S/A selection strategy, one can choose, on average a constant number of sensors and actuators at each time, to approximate the Hankel singular values of the system, while sparsifying the sensor and actuator sets. One of our main contributions is to show that the classical time-varying joint S/A scheduling problem (originally studied by \cite{ATHANS1972397}), can be solved via random sampling. We also propose an alternative to submodularity-based methods and instead use recent advances in theoretical computer science.  

More importantly, we prove that a separation principle holds for the problem of jointly sparsifying the sensor and actuator set with performance guarantees. We show that the joint S/A scheduling problem can be divided into two separate problems:  the sparse sensor schedule and the sparse actuator schedule.

A preliminary version of some of our results in this article submitted for possible publication in a conference proceeding \citep{CDC2019-Milad}; however, their proofs are presented here for the first time. The manuscript also contains several new results including numerical examples, figures, tables, and proofs.

\allowdisplaybreaks

\section{Preliminaries and Definitions}
\subsection{Mathematical Notations}
\label{sec:0}
Throughout the paper, discrete time index is denoted by $k$. The sets of real (integer),  and non-negative real (integer) are represented by $\R$ ($\mathbb Z$), and $\R_+$ ($\mathbb Z_+$), respectively. The set of natural numbers $\{i \in \Z ~:~1 \leq i \leq n\}$ is denoted by $[n]$. The cardinality of a set $\sigma$ is denoted by $\card(\sigma)$. Capital letters, such as $A$ or $B$, stand for real-valued matrices.
The $n$-by-$n$ identity matrix is denoted by $I$. Notation $A \preceq B$ is equivalent to matrix $B-A$ being positive semi-definite. The $n$ eigenvalues of $A \in \R^{n \times n}$ are shown by $\lambda_1(A), \lambda_2(A), \cdots, \lambda_n(A)$. $\lambda_{\max} (.)$ and $\sigma_{\max} (.)$ show the largest eigenvalue and singularvalue of a matrix, respectively. The transpose of matrix $A$ is denoted by $A^\top$. 
The rank of matrix $A$ is referred to by $\Rank(A)$. 

\subsection{Linear Systems, Gramian and Hankel Matrices}

We start with the canonical linear discrete-time, time-invariant dynamics 
\begin{eqnarray}
&& x(k+1) ~=~A \, x(k) ~+~B \, u(k),\label{model-aa}\\
&&y(k) ~=~ C x(k),
\label{model-a}
\end{eqnarray}
where $A \in \R^{n \times n}$, $B \in \R^{n \times m}$, $C \in \R^{p \times n}$ and $k \in \Z_+$.
The state matrix $A$ describes the underlying structure of the system and the interaction strength between the agents, input matrix $B$ represents how the control input enters the system, and output matrix $C$ shows how output vector $y$ relates to the state vector.

The controllability and observability matrices at time $t$ (where $t \geq n$) are given by
\begin{equation}
    \mathcal R(t)~=~\left [ B ~AB~A^2B~\cdots~ A^{t-1}B \right],
    \label{control-matrix}
\end{equation} 
and
\begin{equation}
    \mathcal O(t)~=~\begin{bmatrix}C\\
CA\\
CA^2\\
\cdots\\
CA^{t-1}\end{bmatrix},
    \label{observ-matrix}
\end{equation} 
respectively. 
 It is well-known that  from a numerical standpoint it is better to characterize controllability and observability in terms of the Gramian matrices at time $t$ defined as follows:
 \begin{equation}
\mathcal P(t) ~=~ \sum_{i=0}^{t-1} A^{i} BB^{\top} (A^{i})^\top ~=~ \mathcal R (t) \, \mathcal R^\top(t),
\label{gramian-1}
 \end{equation}
and
  \begin{equation}
\mathcal Q(t) ~=~ \sum_{i=0}^{t-1} (A^{i})^\top C^\top C A^{i} ~=~ \mathcal O^\top(t) \, \mathcal O(t).
\label{gramian-2}
 \end{equation}
 
 \begin{assumption}
\label{assum-1}
Throughout the paper, we assume that the system \eqref{model-aa}-\eqref{model-a} is an $n$-state minimal realization (i.e., the reachability and controllability matrices have full row rank). However, all results presented in this paper can be modified/extended to uncontrollable and unobservable systems.
\end{assumption}

For given linear systems \eqref{model-aa}-\eqref{model-a}, the Hankel matrix is defined as the doubly infinite matrix
\begin{eqnarray*}
\tiny
\mathcal H &=& \begin{bmatrix}
H_1& H_2& H_3& \cdots \\
H_2& H_3& H_4& \cdots \\
H_3& H_4& H_5& \cdots\\
\vdots&\vdots&\vdots&\ddots
\end{bmatrix} \\&=& \begin{bmatrix}C\\
CA\\
CA^2\\
\vdots
\end{bmatrix}\begin{bmatrix}
B&AB&A^2B&\cdots~
\end{bmatrix}\\ &=& \mathcal O \mathcal R,
\end{eqnarray*}
where $H_k = CA^{k-1}B$. The Hankel matrix can be viewed as a mapping between the past inputs and future outputs via the initial state $x(0)$.
Since $\mathcal H = \mathcal O \mathcal R$ and due to Assumption \ref{assum-1}, it follows that $\text{rank} (\mathcal H) = n$. The $n$ nonzero singular values of $\mathcal H$ can be computed by solving two Lyapunov equations (for controllability and observability Gramians) as follows
\begin{eqnarray*}
 \sigma_i(\mathcal H)  &=& \sqrt{\lambda_i(\mathcal H \mathcal H^\top)} \\
 &=& \sigma_i (\mathcal P\mathcal Q) \\
 &=& \lambda_i (\mathcal Q^{\frac{1}{2}}\mathcal P\mathcal Q^{\frac{1}{2}}).
 \label{h-singular}
 \end{eqnarray*}
 The Hankel matrix has a special structure: the elements (blocks) in lines parallel to the anti-diagonal are identical.
 It is well-known that the singular values of the Hankel matrix of a linear system are fundamental invariants of the system, denoting the most controllable and observable modes \citep{antoulas2005approximation}.
It is well known that the states corresponding to small nonzero Hankel singular values are difficult\footnote{The ``difficulty of controllability" of the system can be considered as the energy involved in moving the system from the origin to a uniformly random point on the unit sphere. It is well-known that this quantity can be characterized in terms of controllability Gramian. Moreover, the ``difficulty of observability" of the system can be considered as  ``how observable" the initial state is, over observation horizon. This quantity is closely related to the covariance of the estimate errors in the standard linear least squares problem and can be obtained in terms of observability Gramian.} to control and observe at the same time.

The Hankel norm gives the $L_2$-gain from past inputs to future outputs, and measures the extent to which past inputs effect future outputs of the system.
If the input $u(k)=0$ for $k \geq 0$ and the output is $y(k)$, then the Hankel norm is given by
{\begin{eqnarray*}
 \|G(z)\|_H &:=& \sup_{u \in \mathcal L_2 (-\infty,0)} \frac{\sum_{t=1}^\infty |y(t)|^2}{ \sum_{t=1}^{\infty} |u(-t)|^2}\\
 &=& \sqrt{\lambda_{\max}(\mathcal P \mathcal Q)} ~=~ \sigma_{\max} (\mathcal H),
 \end{eqnarray*}
 where $G(z)$ is a transfer function of dynamics \eqref{model-aa}-\eqref{model-a}, and $\mathcal L_2 (-\infty,0)$ is the space of square summable vector sequences in the interval $(-\infty, 0)$ (which means $\sum_{t=1}^{\infty} |u(-t)|^2 \leq \infty$).  }
In this work, we focus on the time-$t$ Hankel matrix
\begin{eqnarray*}
\tiny
\mathcal H(t) &=& \begin{bmatrix}
H_1& H_2& H_3& \cdots& H_t \\
H_2& H_3& H_4& \cdots& H_{t+1} \\
H_3& H_4& H_5& \cdots& H_{t+2}\\
&&\cdots&&\\
H_t&H_{t+1}& H_{t+2}& \cdots& H_{2t-1}
\end{bmatrix} \\&=& \begin{bmatrix}C\\
CA\\
CA^2\\
\cdots\\
CA^{t-1}\end{bmatrix}\begin{bmatrix}
B&AB&A^2B&\cdots&A^{t-1}B
\end{bmatrix} \\
&=&  \mathcal O(t)\mathcal R(t).
\end{eqnarray*}
Particularly, we have
\begin{eqnarray*}
 \sigma_i(\mathcal H(t))  &=& \sqrt{\lambda_i(\mathcal H(t) \mathcal H^\top(t))} = \sigma_i (\mathcal P(t)\mathcal Q(t)) \\
 &=& \lambda_i (\mathcal Q^{\frac{1}{2}}(t)\mathcal P(t)\mathcal Q^{\frac{1}{2}}(t)).
 \end{eqnarray*}

\begin{remark}
One way to lower the computational complexity of simulations of large-scale dynamical systems is finding a reduced-order model. 
 A common technique for model order reduction is the optimal Hankel-norm approximation. This method provides the best approximation of the original system in the Hankel semi-norm and received significant attention and related development in the 1980s \citep{glover1987model}.
The corresponding state-space realization is the balanced realization where $\tilde{\mathcal P}=\tilde {\mathcal Q} = \diag (\sigma_1, \cdots, \sigma_n)$ as proposed \cite{1102568}. 
In standard model reduction, first we obtain the balanced realization, and then the least observable and controllable modes are truncated.  However, for sparse S/A schedule, we sparsify inputs and outputs in space and time (the number of states does not change) and we utilize a different canonical state realization (see Section \ref{sec:weighted}). 
\end{remark}

{\subsection{Hankel-based Performance Metrics}

Similar to the {\it systemic} notions introduced in \cite{siami2017abstraction, SiamiArXiv18}, we define various performance metrics that capture both controllability and observability properties of the system. These measures are non-negative real-valued operators defined on  the set of all linear dynamical systems governed by  \eqref{model-aa}-\eqref{model-a} and quantify various measures of the performance.   All  of the metrics depend on the symmetric combination of Gramians (\ie ~$\mathcal Q^{\frac{1}{2}}(t) \mathcal P(t) \mathcal Q^{\frac{1}{2}}(t)$) which is a positive definite matrix. Therefore, one can define a systemic Hankel-based performance measure as an operator on the set of Gramian matrices of all $n$-state minimal realization systems, which we represent by $\mathbb S_+^n$.\footnote{The positive-semidefinite cone is denoted by $\mathbb S_+^n$.} We denote the Hankel-based Performance Metrics by $\rho: \mathbb S_+^n \rightarrow \mathbb R$. For many popular choices of $\rho$, one can see that they satisfy the following properties 

	\noindent {(i)} {\it Homogeneity:} for  all $\kappa >1$, 
	\[ \rhoo (\kappa A)~=~\kappa \rhoo (A);\]
	\noindent {(ii)} {\it Monotonicity:} if $A \preceq B$, then
				\[\rhoo (A) ~\leq~ \rhoo (B);\]
\noindent and we call them systemic. For example, the squared Hankel-norm of the system at time $t$ which is defined by 
\[ \rho(\mathcal Q^{\frac{1}{2}}(t) \mathcal P(t) \mathcal Q^{\frac{1}{2}}(t)) := \lambda_{\max}(\mathcal Q^{\frac{1}{2}}(t) \mathcal P(t) \mathcal Q^{\frac{1}{2}}(t)),\]
is systemic. We note that similar criteria have been developed in the experiment design literature \citep{ravi2016experimental,kempthorne1952design,allen-zhu17e}.}

\section{Matrix Reconstruction and Sparsification}
\label{sec:spars}

{The key idea in \cite{SiamiACC18} and \cite{SiamiArXiv18} is to approximate the time-$t$ controllability Gramian as a sparse sum of rank-$1$ matrices, while controlling the approximation error. To this end, a key lemma is used in  \cite{SiamiArXiv18} from the sparsification literature \citep{Christos} to find sparse actuator or sensor schedules. However, in the present work, we are interested in designing a {\em joint} sparse schedule for both sensor and actuator sets; for this, we need to modify a key lemma, known as the Dual Set Lemma in \cite{Christos} to approximate the time-$t$ Hankel singular values. 

Our main result in this section shows how we can handle two sparse subsets with nonidentical indices. 
We then use this result later to design a deterministic algorithm for a joint sparse S/A schedule. More specifically, we need to control the singular values of the product of two matrices which can be written as the symmetrized combination of the two matrices (see Section \ref{sec:weighted}). Each one of these matrices is a sparse sum of rank-$1$ matrices and they reflect controllability and observability properties of the chosen sparse S/A set.

Theorem \ref{th::1} and Algorithm \ref{al::dualsetNew} formalize the procedure of iteratively adding one vector at a time and forming two Gramian matrices.\\

\begin{theorem}
\label{th::1}
Let $V=\{v_1, \ldots, v_{t_1}\}$ and $U=\{ u_1, \ldots, u_{t_2}\}$ such that $\sum_{i=1}^{t_1} v_i v_i^\top = X$ and $\sum_{i=1}^{t_2} u_i u_i^\top = I_n$ where $v_i, u_i \in \R^n$ ($n < t_1, t_2$). Given integer numbers $\kappa_1$ and $\kappa_2$ with $n < \kappa_1 \leq  t_1$ and $n< \kappa_2 \leq t_2$, Algorithm \ref{al::dualsetNew} computes a set of weights $s_i \geq 0$ and $r_i \geq 0$, such that
\begin{eqnarray*}
&& \left( \sum_{i=1}^{t_1} s_i v_i v_i^\top \right)^{\frac{1}{2}} \left( \sum_{i=1}^{t_2} r_i u_i u_i^\top \right)  \left( \sum_{i=1}^{t_1} s_i v_i v_i^\top \right)^{\frac{1}{2}}   \\
&&~~~~~~~~~~~~~~\succeq~ {\rm e}^{-(\epsilon_1+\epsilon_2)} X, 
\end{eqnarray*}
\begin{eqnarray*}
&& \left( \sum_{i=1}^{t_1} s_i v_i v_i^\top \right)^{\frac{1}{2}} \left( \sum_{i=1}^{t_2} r_i u_i u_i^\top \right)  \left( \sum_{i=1}^{t_1} s_i v_i v_i^\top \right)^{\frac{1}{2}}   \\
&&~~~~~~~~~~~~~~\preceq~ {\rm e}^{\epsilon_1+\epsilon_2} X, 
 \end{eqnarray*}
\[ \card \left \{ s_i \neq 0 ~|~ i \in [t_1]\right \}~\leq~ \kappa_1, \]
and
\[ \card \left \{ r_i \neq 0 ~|~ i \in [t_2] \right \}~\leq~ \kappa_2,\]
 where
 \[\epsilon_1  ~:=~ 2 \tanh^{-1}\left (\sqrt{\frac{n}{\kappa_1}}\right),~\text{and}~\epsilon_2  ~:=~ 2 \tanh^{-1}\left (\sqrt{\frac{n}{\kappa_2}}\right).\]
\end{theorem}

Due to space limitations, we refer the interested readers to \citep{Christos} for more details on Algorithm \ref{al::dualsetNew}. However, roughly speaking, the algorithm is based on  choosing vectors in a greedy fashion that satisfy a set of desired properties at each step, leading to bounds on Hankel singular values.
We first define two barriers or potential functions as follows:
 \begin{equation}
  \underline{\phi} (\underline \mu , \underline{\mathcal A}) = \sum_{i=1}^n \frac{1}{\lambda_i(\underline{\mathcal A})-\underline \mu },
  \label{eq::400}
  \end{equation}
 and 
\begin{equation}
 \bar \phi (\bar \mu,\bar{\mathcal A}) ~=~ \sum_{i=1}^n \frac{1}{\bar \mu -\lambda_i(\bar{\mathcal A})}.
 \label{eq::401}
 \end{equation}
These potential functions quantify how far the eigenvalues of $\underline{\mathcal A}$ and $\bar{\mathcal A}$ are from the barriers $\underline \mu$ and $\bar \mu$. These potential functions blow up as any eigenvalue nears the barriers; moreover, they show the locations of all the eigenvalues concurrently.
We then define two parameters $\mathfrak L$ and $\mathfrak U$ as follows:
\begin{eqnarray*}
&&\hspace{-.1cm}\mathfrak L(v, \underline \delta,\underline{\mathcal A},\underline \mu)~= \nonumber \\
&&~~ \frac{ v^\top \left (\underline{\mathcal A}-(\underline \mu +\underline \delta)I_n \right )^{-2}v}{\underline \phi (\underline \mu+\underline \delta, \underline{\mathcal A}) - \underline \phi (\underline \mu, \underline{\mathcal A})}- v^\top \left (\underline{\mathcal A}-(\underline \mu +\underline \delta)I_n \right)^{-1}v,
\end{eqnarray*}
and
\begin{eqnarray*}
&&\hspace{-.1cm}\mathfrak U(u, \bar \delta,\bar{\mathcal A},\bar \mu ) = \nonumber \\
&&~~ \frac{ u^\top ((\bar \mu +\bar \delta)I_n - \bar{\mathcal A})^{-2}u}{\bar \phi (\bar \mu, \bar{\mathcal A}) - \bar \phi (\bar \mu+\bar \delta, \bar{\mathcal A})}+u^\top \left ((\bar \mu +\bar \delta)I_n - \bar{\mathcal A} \right )^{-1}u.
\end{eqnarray*}

The Sherman-Morrison-Woodbury formula inspires the structure of the above quantities for more details on the barrier method see \citep{Christos}. The potential functions \eqref{eq::400} and \eqref{eq::401} are chosen to guide the selection of vectors and scalings at each step $\tau$ and to ensure steady progress of the algorithm.
Small values of these potentials indicate that the eigenvalues of $\bar{\mathcal A}$ and $\underline{\mathcal A}$ do not gather near $\bar{\mu}$ and $\underline{\mu}$, respectively. 
At each iteration, we increase the upper barrier $\bar{\mu}$ by a fixed constant $\bar{\delta}$ and the lower barrier $\bar{\mu}$ by another fixed constant $\underline{\delta}$. It can be shown that as long as the potentials remain bounded, there must exist (at every step $\tau$) a choice of an index $j$ and weights $s_j$ and $r_j$ so that the addition of the associated rank-1 matrices to $\bar{\mathcal A}$ and $\underline{\mathcal A}$, and the increments of barriers do not increase either potential and keep all the eigenvalues of the updated matrix between the barriers (see Algorithm \ref{al::dualsetNew}). Repeating these steps ensures steady growth of all the eigenvalues and yields the desired result.}

This algorithm is tailored from an algorithm from \cite{Christos} (which is deterministic and requires at most $\mathcal O\left( (\kappa_1 t_1+\kappa_2 t_2) n^2\right)$) steps for joint sparse S/A selections. 
We view this algorithm as a subroutine acting on sets $U$ and $V$ as 
\[ s, r = \GenDualSet (V, U, \kappa_1, \kappa_2).\]

\begin{algorithm}
    \caption{A Modified Dual Set Spectral Sparsification $\GenDualSet(V, U, \kappa_1, \kappa_2)$.}
    \label{al::dualsetNew}
\begin{algorithmic}[1]
\Statex {\bf{Input}} $V=\left [v_1,\ldots, v_{t_1}\right] \in \R^{n \times t_1} $, {with} $VV^\top=X$,
    $U=\left [u_1,\ldots, u_{t_2}\right] \in \R^{n \times t_2}$, with $UU^\top=I_n$,
    $\kappa_1 \in \Z_+$, with $n<\kappa_1\leq t_1$,
    $\kappa_2 \in \Z_+$, with $n<\kappa_2\leq t_2$
\Statex  {\bf{Output}} $s=[s_1,s_2, \ldots, s_{t_1}] \in \R^{1 \times t_1}_+$ with $\|s\|_{0} \leq \kappa_1$,
    $r=[r_1,r_2, \ldots, r_{t_2}] \in \R^{1 \times t_2}_+$ with $\|r\|_{0} \leq \kappa_2$,
    
\State Set $s(0)=0_{t_1 \times 1}$, $\underline{\mathcal A}(0)=\bar{\mathcal A}(0)=0_{n \times n}$, $\underline \delta=1$, $\bar \delta=\frac{1+\sqrt{\frac{n}{\kappa_1}}}{1-\sqrt{\frac{n}{\kappa_1}}}$
\For{$\tau=0:\kappa_1-1$}        
\State $\underline \mu({\tau})= \tau - \sqrt{\kappa_1 n}$
\State $\bar \mu({\tau})=\bar \delta \left( \tau+\sqrt{\kappa_1 n}\right)$
\State Find an index $j$ such that
\[ \mathfrak U(X^{-\frac{1}{2}}v_j, \bar \delta, \bar{\mathcal A}(\tau), \bar \mu({\tau})) \leq \mathfrak L(X^{-\frac{1}{2}}v_j, \underline \delta, \underline{\mathcal A}({\tau}), \underline \mu({\tau}))\]
\State Set 
\[\scalebox{0.8}{$\Delta ={2}{\left(\mathfrak U(X^{-\frac{1}{2}} v_j, \bar \delta, \bar{\mathcal A}(\tau), \bar \mu({\tau})) + \mathfrak L(X^{-\frac{1}{2}} v_j, \underline \delta, \underline{\mathcal A}({\tau}), \underline \mu({\tau}))\right)^{-1}}$}\]
\State Update the $j$-th component of $s(\tau)$:~$s({\tau+1})=s(\tau)+ \Delta  \text{e}_j$,
\State $\underline{\mathcal A}({\tau+1}
)=\underline{\mathcal A}({\tau})+\Delta  X^{-\frac{1}{2}} v_jv_j^\top X^{-\frac{1}{2}}$
\State 
$\bar{\mathcal A}({\tau+1})=\bar{\mathcal A}({\tau})+ \Delta X^{-\frac{1}{2}} u_ju_j^\top X^{-\frac{1}{2}}$
\EndFor
\State Set $r(0)=0_{t_2 \times 1}$, $\underline{\mathcal A}(0)=\bar{\mathcal A}(0)=0_{n \times n}$, $\underline \delta=1$, $\bar \delta=\frac{1+\sqrt{\frac{n}{\kappa_2}}}{1-\sqrt{\frac{n}{\kappa_2}}}$
        \vspace{.1cm}
\For{$\tau=0:\kappa_2-1$}        
\State $\underline \mu({\tau})= \tau - \sqrt{\kappa_2 n}$
\State $\bar \mu({\tau})=\bar \delta \left( \tau+\sqrt{\kappa_2 n}\right)$
\State Find an index $j$ such that
\[ \mathfrak U(u_j, \bar \delta, \bar{\mathcal A}(\tau), \bar \mu({\tau})) \leq \mathfrak L(u_j, \underline \delta, \underline{\mathcal A}({\tau}), \underline \mu({\tau}))\]
\State Set 
\[\Delta ={2}{\left(\mathfrak U(u_j, \bar \delta, \bar{\mathcal A}(\tau), \bar \mu({\tau})) + \mathfrak L(u_j, \underline \delta, \underline{\mathcal A}({\tau}), \underline \mu({\tau}))\right)^{-1}}\]
\State Update the $j$-th component of $r(\tau)$:
\[r({\tau+1})=r(\tau)+ \Delta  \text{e}_j,\]
\State $\underline{\mathcal A}({\tau+1}
)=\underline{\mathcal A}({\tau})+\Delta  u_ju_j^\top$
\State $\bar{\mathcal A}({\tau+1})=\bar{\mathcal A}({\tau})+ \Delta u_ju_j^\top$
\EndFor

\State{\bf Return}{\[s=\kappa _1^{-1}\left(1- \sqrt{\frac{n}{\kappa_1}}\right) s(\kappa_1)\]
\[r=\kappa_2 ^{-1}\left(1- \sqrt{\frac{n}{\kappa_2}}\right) r(\kappa_2)\]}
\end{algorithmic}
\end{algorithm}

We now present the proof of Theorem \ref{th::1}.\\

\begin{proof}
{To prove this theorem we first use \citep[Lemma 1]{Christos}. We first define an isotropic set of $t_1$ vectors based on set $V=\{v_1, \ldots, v_{t_1}\}$ as follows
\begin{equation}
\bar V = \{ \bar v_i = X^{-\frac{1}{2}} v_i | ~v_i \in V\}.
\label{eq::464}
\end{equation}
Using \eqref{eq::464}, we have
\begin{equation}
\sum_{i=1}^{t_1} \bar v_i \bar v_i^\top = I_n.
\label{eq:468}
\end{equation}
Then, according to the Dual Set Lemma in \cite{Christos} and Line 1 to Line 10 of Algorithm \ref{al::dualsetNew} where $\bar s = \kappa _1^{-1}\left(1- \sqrt{\frac{n}{\kappa_1}}\right) s(\kappa_1)$, we get
\begin{equation}
\left( 1- \sqrt{\frac{n}{\kappa_1}}\right)^2 I_n ~\preceq~ \sum_{i=1}^{t_1} \bar s_i \bar v_i \bar v_i^\top   ~\preceq~ \left( 1+\sqrt{\frac{n}{\kappa_1}}\right)^2 I_n,
\label{eq:474}
\end{equation}
where $\card \{ s_i \neq 0 |~i \in [t_1] \} \leq \kappa_1$. This can be rewritten as follows
\begin{equation}
{\rm e}^{-\epsilon_1} I_n ~\preceq~ \sum_{i=1}^{t_1} s_i \bar v_i \bar v_i^\top   ~\preceq~ {\rm e}^{-\epsilon_1} I_n,
\label{eq:474A}
\end{equation}
where $\epsilon_1=2 \tanh^{-1}\left (\sqrt{\frac{n}{\kappa_1}}\right)$, and 
\[ s = \kappa _1^{-1}\left(1+ \sqrt{\frac{n}{\kappa_1}}\right)^{-1} s(\kappa_1),\]
where $s(.)$ is defined in Algorithm \ref{al::dualsetNew} and $s=[s_1,s_2, \cdots, s_{t_1}]$.
Next, based on \eqref{eq:468},  \eqref{eq:474}, and $\sum_{i=1}^{t_1} v_i v_i^\top = X$, we get 
\begin{equation}
{\rm e}^{-\epsilon_1}  X ~\preceq~ \sum_{i=1}^{t_1} s_i  v_i  v_i^\top   ~\preceq~ {\rm e}^{\epsilon_1}  X,
\label{eq:475}
\end{equation}
where \eqref{eq:475} is obtained by multiplying positive definite matrix $X^{\frac{1}{2}}$ from both sides of \eqref{eq:474}.
Similarly, according to the Dual Set Lemma in \cite{Christos} and Line 11 to Line 21 of Algorithm \ref{al::dualsetNew} where $\bar r = \kappa _2^{-1}\left(1- \sqrt{\frac{n}{\kappa_2}}\right) r(\kappa_2)$, we get
\begin{equation}
\left( 1- \sqrt{\frac{n}{\kappa_2}}\right)^2 I_n ~\preceq~ \sum_{i=1}^{t_2} \bar r_i u_i u_i^\top ~\preceq~\left( 1+ \sqrt{\frac{n}{\kappa_2}}\right)^2 I_n,
\label{eq::481}
\end{equation}
where  $\card \{ r_i \neq 0 |~i \in [t_2] \} \leq \kappa_2$. This can be rewritten as follows
\begin{equation}
{\rm e}^{-\epsilon_2} I_n ~\preceq~  \sum_{i=1}^{t_2} r_i u_i u_i^\top    ~\preceq~ {\rm e}^{-\epsilon_2} I_n,
\label{eq:474A}
\end{equation}
where $\epsilon_2=2 \tanh^{-1}\left (\sqrt{\frac{n}{\kappa_2}}\right)$, and 
\[ r = \kappa _2^{-1}\left(1+ \sqrt{\frac{n}{\kappa_2}}\right)^{-1} r(\kappa_2),\]
where $r(.)$ is defined in Algorithm \ref{al::dualsetNew} and $r=[r_1,r_2, \cdots, r_{t_2}]$.
Using \eqref{eq::481} and the fact that $\left( \sum_{i=1}^{t_1} s_i v_i v_i^\top \right)^{\frac{1}{2}}\succeq 0 $, it follows
\begin{eqnarray}
&&{\rm e}^{-\epsilon_2} \sum_{i=1}^{t_1} s_i v_i v_i^\top ~\preceq~~~~~~~~~~~~~~~~~~~~~~~~~~~~~~~\nonumber \\
 &&\left( \sum_{i=1}^{t_1} s_i v_i v_i^\top \right)^{\frac{1}{2}} \left( \sum_{i=1}^{t_2} r_i u_i u_i^\top \right)  \left( \sum_{i=1}^{t_1} s_i v_i v_i^\top \right)^{\frac{1}{2}} \nonumber \\  
 &&\preceq~{\rm e}^{\epsilon_2}  \sum_{i=1}^{t_1} s_i v_i v_i^\top.
\label{eq::589}
\end{eqnarray}
Finally combining \eqref{eq:475} and \eqref{eq::589}, we get the desired results.}
\end{proof}

In the next section, we show how various Hankel-based measures can be approximated by selecting a sparse set of actuators and sensors.

\section{Joint Sparse S/A Scheduling Problems}
\label{sec:sparse}
For given linear system \eqref{model-aa}-\eqref{model-a} with a general underlying structure, the joint S/A scheduling problem seeks to construct a schedule of the control inputs and sensor outputs that keeps the number of active actuators and sensors much less than the fully sensed/actuated system such that the Hankel-based performance matrices of the original and the new systems are similar in an appropriately defined sense.
Specifically, given a canonical linear, time-invariant system \eqref{model-aa}-\eqref{model-a} with $m$ actuators, $p$ sensors and Gramians $\mathcal P(t)$, $\mathcal Q(t)$ at time $t$, our goal is to find a joint sparse S/A schedule such that the resulting system with Hankel matrix $\mathcal H_s(t)$ is well-approximated, i.e., 
\begin{equation}
\left| \log \frac{\rho\left (\mathcal Q^{\frac{1}{2}}(t) \mathcal P(t) \mathcal Q^{\frac{1}{2}}(t)\right) }{\rho\left (\mathcal Q_s^{\frac{1}{2}}(t) \mathcal P_s(t) \mathcal Q_s^{\frac{1}{2}}(t)\right )} \right| ~\le~ \epsilon,
\label{518}
\end{equation}
where $\rho$ is any systemic performance metric that quantifies the performance of the system for example as the $\mathcal L_2$-gain from past inputs to future outputs, and $\epsilon \geq 0 $ is the approximation factor. The systemic performance metrics are defined based on the Hankel singular values, and we will show that ``close" controllability and observability Gramian matrices result in approximately the same values.
Our goal here is to answer the following questions:
(1) What are the minimum numbers of actuators and sensors that need to be chosen to achieve a good approximation of the system where the full sets of actuators and sensors utilized? (2) What is the relation between the numbers of selected actuators and sensors and performance loss?
(3) Does a sparse approximation schedule exist with at most constant numbers of  active actuators and sensors at each time?
(4) What is the time complexity of choosing the subsets of actuators and sensors with guaranteed performance bounds?

In the rest of this paper, we show how some fairly recent advances in theoretical computer science can be utilized to answer these questions. Recently, Marcus, Spielman, and Srivastava introduced a new variant of the probabilistic method which ends up solving the so-called Kadison-Singer (KS) conjecture \citep{marcus}.  We use the solution approach to the KS conjecture together with  a combination of  tools from  Sections \ref{sec:sparse} and \ref{sec:weighted}  to find a sparse approximation of the sparse S/A scheduling problem with algorithms that have favorable time-complexity.

\section{A Weighted Joint Sparse S/A Schedule}
\label{sec:weighted}

As a starting point, we allow for scaling of the selected input and output signals while keeping the input and output scaling bounded. The input and output scalings allow for an extra degree of freedom that could allow for further sparsification of the sensor/actuator set. 
Given \eqref{model-aa}-\eqref{model-a}, we define a weighted, joint sensor and actuator schedule by 
\[ \sigma ~=~ \left( \{\sigma^{\text{(s)}}_k\}_{k=0}^{t-1}, \{\sigma^{\text{(a)}}_k\}_{k=0}^{t-1}\right),\]
where 
\[ \sigma^{\text{(s)}}_k= \left \{i |\,  \mathfrak s_i(k) >0, ~i \in [p] \right \} \subseteq [p],\]
\[ \sigma^{\text{(a)}}_k=\left \{i |\,  \mathfrak a_i(k) >0, ~i \in [m] \right\} \subseteq [m], \]
 and non-negative input and output scalings (i.e.,  $\mathfrak a_i(k) \geq 0$, $\mathfrak s_i(k) \geq 0$). 
The resulting system with this schedule is 
\begin{eqnarray}
 x(k+1) &=& A \, x(k) \,+\, \sum_{i \in \sigma^{\text{(a)}}_k} \mathfrak a_{i}(k)\, b_i \, u_i(k),\label{model:bb}\\
 y(k) &=&  \sum_{i \in \sigma^{\text{(s)}}_k} \mathfrak s_{i}(k)\, {\text e}_i\, c_i \, x(k),
 \label{model:b}
 \end{eqnarray}
where $b_i$'s are columns of matrix $B \in \R^{n \times m}$, $c_i$'s are rows of matrix $C \in \R^{p \times n}$, and $\text{e}_i$'s are the standard basis for $\R^p$;
scaling $\mathfrak a_{i}(k) \geq 0$ shows the strength of the $i$-th control input at time $k$; and similarly $\mathfrak s_{i}(k) \geq 0$ shows the strength of the $i$-th output at time $k$. Equivalently, the dynamics can be rewritten as 
\begin{eqnarray}
 && x(k+1) ~=~A \, x(k) ~+~ \mathcal B(k) \, u(k),
 \label{g:model-0}\\
 && y(k)~=~\mathcal C(k) \, x(k),
\label{g:model}
\end{eqnarray}
with time-varying input and output matrices 
\begin{equation*}
\mathcal B(k) ~= ~B \, \Lambda(k), 
\label{B-matrix}
\end{equation*} 
and 
\begin{equation*}
\mathcal C(k) ~=~ \Gamma (k) \, C,
\label{C-matrix}
\end{equation*}
 where $\Lambda(k)$ and $\Gamma(k)$ are diagonal, and their nonzero diagonal entries show selected actuators and sensors at time $k$, which means
\[ \Lambda(k)= \diag \left(\mathfrak a_1(k), \cdots, \mathfrak a_m(k) \right),\]
and
\[ \Gamma(k)= \diag \left(\mathfrak s_1(k), \cdots, \mathfrak s_p(k)\right). \]

The controllability Gramian and observability Gramian at time $t$ for this system can be rewritten as
\begin{equation}
\mathcal P_{s}(t) ~=~ \sum_{k=0}^{t-1} \sum_{j \in \sigma^{\text{(a)}}_k} \mathfrak a_j^2(k)\left (A^{t-k-1} b_j\right)\left (A^{t-k-1}b_j\right)^\top,
\label{W_s}
\end{equation}
and
\begin{equation}
\mathcal Q_{s}(t) ~=~ \sum_{k=0}^{t-1} \sum_{j \in \sigma^{\text{(s)}}_k} \mathfrak s_j^2(k)\left (c_j A^{t-k-1}\right)^\top\left (c_j A^{t-k-1}\right).
\label{Q_s}
\end{equation}

Our goal is to keep the numbers of active sensors and actuators on average less than $d_{\text{s}}$ and $d_{\text{a}}$, i.e., 
\begin{equation}
\frac{ \sum_{k=0}^{t-1}\card\left \{\sigma^{\text{(s)}}_k\right \}}{t} ~\leq~  d_{\text{s}},
 \label{def-d-s}
 \end{equation}
 and
\begin{equation}
\frac{ \sum_{k=0}^{t-1}\card\left \{\sigma^{\text{(a)}}_k\right \}}{t} ~\leq~  d_{\text{a}},
 \label{def-d-a}
\end{equation}
such that the Hankel matrix of the fully actuated/sensed system, is ``close" to the Hankel matrix of the new sparsely actuated/sensed system. Of course, this approximation will require horizon lengths that are potentially longer than the dimension of the state. 

\begin{assumption}
Throughout this paper, we assume the horizon length is fixed and is given by $t \geq n$. 
\end{assumption}

The definition below formalizes the meaning of approximation.

\begin{definition}[$(\epsilon,d_{\text{s}} )$-sensor schedule]
\label{def:obs}
We call system \eqref{model:bb}-\eqref{model:b} the $(\epsilon, d_{\text{s}} )$-sensor schedule for system \eqref{model-aa}-\eqref{model-a} if and only if
\begin{equation}
{\rm e}^{-\epsilon}\,\mathcal Q(t) ~\preceq ~ \mathcal Q_s(t) ~\preceq ~{\rm e}^{\epsilon}\,\mathcal Q(t),
\label{eq:321-obs}
\end{equation}
where $\mathcal Q(t)$ and $\mathcal Q_s(t)$ are the observability Gramian matrices of systems \eqref{model-aa}-\eqref{model-a} and \eqref{model:bb}-\eqref{model:b}, respectively. Parameter $d_{\text{s}}$ is defined by \eqref{def-d-s} as an upper bound on the average number of active sensors, and $\epsilon \in (0,1)$ is the approximation factor.
\end{definition}
 
Next we define the $(\epsilon,d_{\text{a}})$-actuator schedule for dynamical system \eqref{model-aa}-\eqref{model-a}.
\begin{definition}[$(\epsilon,d_{\text{a}})$-actuator schedule]
\label{def:con}
We call system \eqref{model:bb}-\eqref{model:b} the $(\epsilon, d_{\text{a}})$-actuator schedule of system \eqref{model-aa}-\eqref{model-a} if and only if
\begin{equation}
{\rm e}^{-\epsilon}\,\mathcal P(t) ~\preceq ~ \mathcal P_s(t) ~\preceq ~{\rm e}^{\epsilon}\,\mathcal P(t),
\label{eq:321-con}
\end{equation}
where $\mathcal P(t)$ and $\mathcal P_s(t)$ are the controllability Gramian matrices of \eqref{model-aa}-\eqref{model-a} and \eqref{model:bb}-\eqref{model:b}, respectively, and  parameter $d_{\text{a}}$ is defined by \eqref{def-d-a} as an upper bound on the average number of active actuators, and $\epsilon \in (0,1)$ is the approximation factor.
\end{definition}

\begin{remark}
 While it might appear that allowing for the choice of $\mathfrak s_i(k)$ and $\mathfrak a_i(k)$ might lead to amplification of output and intput signals,  we  note that the scaling cannot be too large because the approximations \eqref{eq:321-obs} and \eqref{eq:321-con} are two-sided. Specifically, by taking  the trace from both sides of \eqref{eq:321-obs} and \eqref{eq:321-con}, we can see that the weighted summations of  $\mathfrak s_i^2(k)$'s and $\mathfrak a_i^2(k)$'s are bounded. Moreover, based on Definitions \ref{def:obs} and \ref{def:con}, the ranks of matrices $\mathcal Q(t)$ and $\mathcal Q_s(t)$ are the same, similarly for matrices $\mathcal P(t)$ and $\mathcal P_s(t)$. Thus, the resulting $(\epsilon,d)$-approximation remains controllable and observable (recall that we assume that the original system is controllable and observable).
\end{remark}

We now define the joint sparse S/A schedule for system \eqref{model-aa}-\eqref{model-a} based on the Hankel singular values of the system.
\begin{definition}[$(\epsilon,d_{\text{s}},d_{\text{a}})$-joint S/A schedule]
\label{def:S/A}
We call system \eqref{model:bb}-\eqref{model:b} the $(\epsilon, d_{\text{s}},d_{\text{a}})$-joint S/A schedule of system \eqref{model-aa}-\eqref{model-a} if and only if
\begin{eqnarray*}
{\rm e}^{-\epsilon}\left( \mathcal Q(t)^{\frac{1}{2}} \mathcal P(t) \mathcal Q^{\frac{1}{2}}(t) \right) &\preceq& \mathcal Q_s^{\frac{1}{2}}(t) \mathcal P_s(t) \mathcal Q_s^{\frac{1}{2}}(t) \\
& \preceq& {\rm e}^{\epsilon} \left( \mathcal Q^{\frac{1}{2}}(t) \mathcal P(t) \mathcal Q^{\frac{1}{2}}(t) \right ),
\end{eqnarray*}
where $\mathcal P$, $\mathcal Q$, $\mathcal P_s$, and $\mathcal Q_s$ are the controllability and observability Gramians of \eqref{model-aa}-\eqref{model-a} and \eqref{model:bb}-\eqref{model:b}, respectively, and  parameters $d_{\text{s}}$ and $d_{\text{a}}$ are upper bounds on the average number of active sensors and actuators, and $\epsilon \in (0,1)$ is the approximation factor.\footnote{We should note that according to Definition \ref{def:S/A} if system \eqref{model:bb}-\eqref{model:b} is the $(\epsilon, d_{\text{s}},d_{\text{a}})$-joint S/A schedule, then it is also $(\epsilon_1, d_{\text{s}},d_{\text{a}})$-joint S/A schedule where $\epsilon_1 \geq \epsilon$.}
\end{definition}
The Hankel singular values can be computed from the reachability and observability Gramians. Note that $\mathcal Q^{\frac{1}{2}}(t) \mathcal P(t) \mathcal Q^{\frac{1}{2}}(t)$ and $\mathcal P(t) \mathcal Q(t)$ share the same eigenvalues. Therefore, the $k$-th largest Hankel singular values of system \eqref{g:model-0}-\eqref{g:model} are bounded from below and above by $\rm e^{\pm \epsilon}$ times the $k$-th largest Hankel singular values of system \eqref{g:model-0}-\eqref{g:model}.\\

\begin{remark}
The Hankel singular values reflect the joint controllability and observability of the balanced states. The Hankel singular values of the fully actuated and sensed system \eqref{model-aa}-\eqref{model-a} are well-approximated by  the Hankel singular values of  the joint S/A schedule.
\end{remark}

\subsubsection*{Construction Results} The next theorem constructs a solution for the sparse weighted S/A scheduling problem in polynomial time.

 \begin{algorithm}
 \caption{ A deterministic greedy-based algorithm to construct a sparse weighted S/A schedule (Theorem \ref{th-main-approx}).}
    \label{al-approx}
    \begin{algorithmic}[1]
\Statex {\bf{Input}} $A \in \R^{n \times n}$, $B \in \R^{n \times m}$, $C \in \R^{p \times n}$, $t$, $d_{\text{s}}$, and $d_{\text{a}}$
\Statex  {\bf{Output}} $\mathfrak s_i(k) \geq 0$ for $(i,k+1) \in [m] \times [t]$, $\mathfrak a_i(k) \geq 0$ for $(i,k+1) \in [p] \times [t]$
\vspace{.4cm}
\State $\mathcal R(t) := \left [ B~AB~A^2B~\cdots~A^{t-1}B \right ]$
\State  $\mathcal O^\top(t) := \left [ C^\top~A^\top C^\top~(A^2)^\top C^\top~\cdots~(A^{t-1})^\top C^\top\right ]^\top$
        \vspace{.1cm}
\State Set $V= \left(\mathcal R(t) \mathcal R^\top(t)\right)^{\frac{1}{2}}\mathcal O^\top(t)$
\State Set $U=\left(\mathcal R(t) \mathcal R^\top(t)\right)^{-\frac{1}{2}}\mathcal R(t)$
\State  Run $s, a = \GenDualSet (V, U, d_{\text{s}}t, d_{\text{a}}t)$
\State  Return $\mathfrak s_i(k):=\sqrt{r_{i+pk}}$ for $(i,k+1) \in [p] \times [t]$, $\mathfrak a_i(k):=\sqrt{s_{i+mk}}$ for $(i,k+1) \in [m] \times [t]$
        \vspace{.1cm}
\end{algorithmic}
\end{algorithm}

\vspace{.5cm}
\begin{theorem}
\label{th-main-approx}
Given the  time horizon $t \geq n$, model \eqref{model-aa}-\eqref{model-a}, $d_{\text a}> 1$ and $d_{\text s} > 1$,  
Algorithm \ref{al-approx} deterministically constructs a joint S/A schedule such that the resulting system \eqref{g:model-0}-\eqref{g:model} is a  $\left (\epsilon, d_{\text s}, d_{\text a} \right )$-approximation of \eqref{model-aa}-\eqref{model-a} with \[\epsilon~=~ \epsilon_{\text{s}} + \epsilon_{\text{a}},\]
 where
\[\epsilon_{\text{s}}  ~=~ 2 \tanh^{-1}\left (\sqrt{\frac{n}{d_{\text{s}}t}}\right),\]
and
\[\epsilon_{\text{a}} ~=~ 2 \tanh^{-1}\left (\sqrt{\frac{n}{d_{\text{a}}t}}\right),\]
 in at most $\mathcal O\left ( (pd_{\text s}+md_{\text a})(tn)^2 \right)$ operations.
\end{theorem}

\begin{proof}
{The observability Gramian  of \eqref{model-aa}-\eqref{model-a} at time $t$ is given by
\begin{equation}
\mathcal Q(t)=\sum_{i=0}^{t-1} \sum_{j=1}^p (\underbrace{A^{i} c_j^\top}_{\bar u_{ij}})(A^{i} c_j^\top)^\top = \sum_{i=0}^{t-1} \sum_{j=1}^p \bar u_{ij} \bar u_{ij}^\top.
\label{eq:4600}
\end{equation}
By multiplying $\mathcal Q^{-\frac{1}{2}}(t)$ on both sides of \eqref{eq:4600}, it follows that 
\begin{eqnarray}
I_n&=&\sum_{i=0}^{t-1} \sum_{j=1}^p \underbrace{(\mathcal Q^{-\frac{1}{2}}(t)A^{i} c_j^\top)}_{u_{ij}} \, (\mathcal Q^{-\frac{1}{2}}(t) A^{i} c_j^\top )^\top \nonumber \\
 &=& \sum_{i=0}^{t-1} \sum_{j=1}^p  u_{ij}  u_{ij}^\top. 
\label{sum::identity}
\end{eqnarray}
In addition, the controllability Gramian  of \eqref{model-aa}-\eqref{model-a} at time $t$ is given by
\begin{equation}
\mathcal P(t)=\sum_{i=0}^{t-1} \sum_{j=1}^m (\underbrace{A^{i} b_j}_{\bar v_{ij}})(A^{i} b_j)^\top = \sum_{i=0}^{t-1} \sum_{j=1}^m \bar v_{ij} \bar v_{ij}^\top.
\label{eq:465}
\end{equation}
By multiplying $\mathcal Q(t)$ on both sides of \eqref{eq:465}, it follows that 
\begin{eqnarray}
\mathcal Q(t) \mathcal P(t) \mathcal Q(t)&=&\sum_{i=0}^{t-1} \sum_{j=1}^m \underbrace{(\mathcal Q(t)A^{i} b_j)}_{v_{ij}} \, (\mathcal Q(t) A^{i} b_j )^\top\nonumber \\ &=& \sum_{i=0}^{t-1} \sum_{j=1}^m  v_{ij}  v_{ij}^\top. 
\label{sum::identity}
\end{eqnarray}
Next, we define $\kappa_1:= d_{\text{a}}t$, $\kappa_2 := d_{\text{s}} t$,
\[V :=\{ v_{ij} | i+1 \in [t], j \in [m]\},\]
and 
\[U := \{ u_{ij} | i+1 \in [t], j \in [p]\}.\] 
 We now apply Theorem \ref{th::1}, which shows that there exist scalars $s_{ij} \geq 0$ and $r_{ij} \geq 0$ such that
\begin{eqnarray}
&& \hspace{-.4cm} \scalebox{.8}{$\left( \sum_{i=0}^{t-1}\sum_{j=1}^{m} s_{ij} v_{ij} v_{ij}^\top \right)^{\frac{1}{2}} \left( \sum_{i=0}^{t-1}\sum_{j=1}^{p} r_{ij} u_{ij} u_{ij}^\top \right)  \left( \sum_{i=0}^{t-1}\sum_{j=1}^{m} s_{ij} v_{ij} v_{ij}^\top \right)^{\frac{1}{2}}$} \nonumber \\
&&~~~~~~~~~~~~~~\succeq~ {\rm e}^{-(\epsilon_{\text a}+\epsilon_{\text s})} \mathcal Q(t) \mathcal P(t) \mathcal Q(t), 
\label{eq::937}
\end{eqnarray}
\begin{eqnarray}
&& \hspace{-.4cm}  \scalebox{.8}{$\left( \sum_{i=0}^{t-1}\sum_{j=1}^{m} s_{ij} v_{ij} v_{ij}^\top \right)^{\frac{1}{2}} \left( \sum_{i=0}^{t-1}\sum_{i=1}^{p} r_{ij} u_{ij} u_{ij}^\top \right)  \left( \sum_{i=0}^{t-1}\sum_{j=1}^{m} s_{ij} v_{ij} v_{ij}^\top \right)^{\frac{1}{2}}$} \nonumber   \\
&&~~~~~~~~~~~~~~\preceq~ {\rm e}^{\epsilon_{\text a}+\epsilon_{\text s}} \mathcal Q(t) \mathcal P(t) \mathcal Q(t), 
\label{eq::943}
 \end{eqnarray}
\[ \card \left \{ s_{ij} \neq 0 ~|~ i+1 \in [t], j \in [m]\right \}~\leq~ \kappa_1, \]
and
\[ \card \left \{ r_{ij} \neq 0 ~|~ i+1 \in [t], j \in [p] \right \}~\leq~ \kappa_2.\]
where
\[\epsilon_{\text{s}}  ~=~ 2 \tanh^{-1}\left (\sqrt{\frac{n}{d_{\text{s}}t}}\right),\]
and
\[\epsilon_{\text{a}} ~=~ 2 \tanh^{-1}\left (\sqrt{\frac{n}{d_{\text{a}}t}}\right).\]
Using \eqref{W_s}, \eqref{Q_s}, \eqref{eq::937}, and \eqref{eq::943}, we get
{\begin{eqnarray}
\hspace{-.3cm}{\rm e}^{-\epsilon}\left( \mathcal Q(t)^{\frac{1}{2}} \mathcal P(t) \mathcal Q^{\frac{1}{2}}(t) \right) &\preceq& \mathcal Q_s^{\frac{1}{2}}(t) \mathcal P_s(t) \mathcal Q_s^{\frac{1}{2}}(t) \nonumber \\
& \preceq& {\rm e}^{\epsilon} \left( \mathcal Q^{\frac{1}{2}}(t) \mathcal P(t) \mathcal Q^{\frac{1}{2}}(t) \right),
\label{eq:963}
\end{eqnarray}}
where $\epsilon~:=~\epsilon_{\text{s}}+\epsilon_{\text{a}}$, and 
\begin{equation}
       \mathfrak s_j(t-i-1):=\sqrt{r_{ij}}, ~\text{and}~ \mathfrak a_j(t-i-1):=\sqrt{s_{ij}}.
\end{equation}
Finally, using \eqref{eq:963}, and Definition \ref{def:S/A}, we obtain the desired result. Moreover, this algorithm  runs in $d_{\text s} t+d_{\text a}t$ iterations; For the first loop of Algorithm \ref{al::dualsetNew}, in each iteration,  the functions $\mathfrak U$ and $\mathfrak L$ are evaluated at most $mt$ times. All $mt$ evaluations for both functions need at most $\mathcal O(n^3+mtn^2)$ time, because for all of them the matrix inversions can be calculated once. Finally, the updating step needs an additional $\mathcal O(n^2)$ time. Similarly for the second loop of Algorithm \ref{al::dualsetNew}, in each iteration,  the functions $\mathfrak U$ and $\mathfrak L$ are evaluated at most $pt$ times. All $pt$ evaluations for both functions need at most $\mathcal O(n^3+ptn^2)$ time, because for all of them the matrix inversions can be calculated once. Finally, the updating step needs an additional $\mathcal O(n^2)$ time.
 Overall, the complexity of the algorithm is of the order  $\mathcal O\left ( (pd_{\text s}+md_{\text a})(tn)^2 \right)$. }
\end{proof}

\vspace{.1cm}

The results presented in this paper also work for the case of linear time-varying systems. Because we can easily form the Gramian matrices for linear time-varying discrete-time systems, and then we can apply the result presented in Theorem \ref{th-main-approx} to obtain sparse actuator/sensor schedules.
 
\subsubsection*{ Tradeoffs}
Theorem \ref{th-main-approx} illustrates a tradeoff between the average numbers of active actuators and sensors (\ie, $d_{\text a}$ and $d_{\text s}$) and the time horizon $t$ (also known as the time-to-control). 
This implies that the reduction in the  average numbers of active actuators and sensors comes at the expense of  increasing time horizon $t$ in order to get the same approximation factor $\epsilon$. 
Moreover, the approximation becomes more accurate as  $t$, $d_{\text s}$, $d_{\text a}$ are increased. Indeed, increasing $d_{\text s}$ and $d_{\text a}$ will require more active sensors and actuators, and larger $t$ requires a larger observation and control time window.

\subsubsection*{The separation principle}

We next state our main theorem that shows that to design a joint sparse S/A schedule, it is enough to design two separate schedules: the sparse sensor schedule and the sparse actuator schedule. Thus, the joint S/A scheduling problem can be broken into two decoupled parts, that facilitates the design process.
\begin{theorem}
\label{theorem3}
To obtain $(\epsilon,d_{\text{s}}, d_{\text{a}})$-joint S/A schedule of linear system \eqref{g:model-0}-\eqref{g:model}, it is sufficient to first design an $(\epsilon
_1,d_{\text{s}})
$-sensor schedule and then  design an $(\epsilon-\epsilon_1,d_{\text{a}})$-actuator schedule with $\epsilon_1 < \epsilon$.
\end{theorem}

\begin{proof}
First, we consider the following state-space realization of system \eqref{model-aa}-\eqref{model-a}, where $\tilde x = T x$ and it follows that
\[ x(t) ~~\rightarrow~~ \tilde x(t) ~=~ T x(t), \]
\[ \mathcal P(t) ~~\rightarrow ~~ \tilde {\mathcal P}(t) ~=~ T^{-1} \mathcal{P}(t) T^{-\top},\]
and $\mathcal {Q}(t) ~~\rightarrow~~ \tilde{\mathcal Q}(t) ~=~ T^{\top} \mathcal{Q}(t) T$. 
Let us assume that $T := \mathcal P^{\frac{1}{2}}(t)$, then the new controllability Gramian matrix is given by $\tilde{\mathcal P}(t) = I_n$. Using the new change of coordinates and our result from \cite{SiamiArXiv18}, we get
\begin{equation}
{\rm e}^{-\epsilon_1}\, I_n ~\preceq~ \tilde{\mathcal P}_s(t)  ~\preceq~ {\rm e}^{\epsilon_1}\,I_n, 
\label{eq:818}
\end{equation}
where $\tilde {\mathcal P}_s(t)$ is the controllability Gramian of the $(\epsilon_1, d_{\text{a}})$- actuator schedule. By multiplying $\tilde{\mathcal Q}^{\frac{1}{2}}_s(t)$ to \eqref{eq:818} from both sides, we have
\begin{equation}
{\rm e}^{-\epsilon_1}\, \tilde{\mathcal Q_s}(t) ~\preceq~ \tilde{\mathcal Q_s}^{\frac{1}{2}}(t) \tilde{\mathcal P}_s(t)  \mathcal Q_s^{\frac{1}{2}}(t)  ~\preceq~ {\rm e}^{\epsilon_1}\,  \tilde{\mathcal Q}_s(t).
\label{eq:826}\end{equation}
Similarly, based on Definition \ref{def:obs} for observability Gramian and \cite[Thm. 1]{SiamiArXiv18} for the dual notion of sensor scheduling, we get
\begin{equation}
{\rm e}^{-\epsilon_2}\,  \tilde{\mathcal Q}(t) ~\preceq~ \tilde{\mathcal Q}_s(t)  ~\preceq~ {\rm e}^{\epsilon_2}\,  \tilde{\mathcal Q}(t), 
\label{eq:830}
\end{equation}
where $\epsilon_2 = \epsilon - \epsilon_1$.
By combining \eqref{eq:826} and \eqref{eq:830}, it follows that
{\begin{equation}
{\rm e}^{-(\epsilon_1+\epsilon_2)}\,  \tilde{\mathcal Q}(t) ~\preceq~ \tilde{\mathcal Q}_s^{\frac{1}{2}}(t) \tilde{\mathcal P}_s(t)  \tilde{\mathcal Q}_s^{\frac{1}{2}}(t)  ~\preceq~ {\rm e}^{\epsilon_1+\epsilon_2}\,  \tilde{\mathcal Q}(t), \end{equation}}
and finally, we have
{\begin{eqnarray}
{\rm e}^{-\epsilon}\left( \mathcal Q^{\frac{1}{2}}(t) \mathcal P(t) \mathcal Q^{\frac{1}{2}}(t) \right) &\preceq& \mathcal Q_s^{\frac{1}{2}}(t) \mathcal P_s(t) \mathcal Q_s^{\frac{1}{2}}(t)\nonumber \\
&\preceq& {\rm e}^{\epsilon}\left( \mathcal Q^{\frac{1}{2}}(t) \mathcal P(t) \mathcal Q^{\frac{1}{2}}(t)\right).
\label{eq:842}
\end{eqnarray}}
Using \eqref{eq:842} and Definition \ref{def:S/A}, we obtain the desired result.
\end{proof}

\begin{figure}
	\centering       
	\includegraphics[trim = 0.5 0.5 0.5 0.5, clip,width=.45 \textwidth]{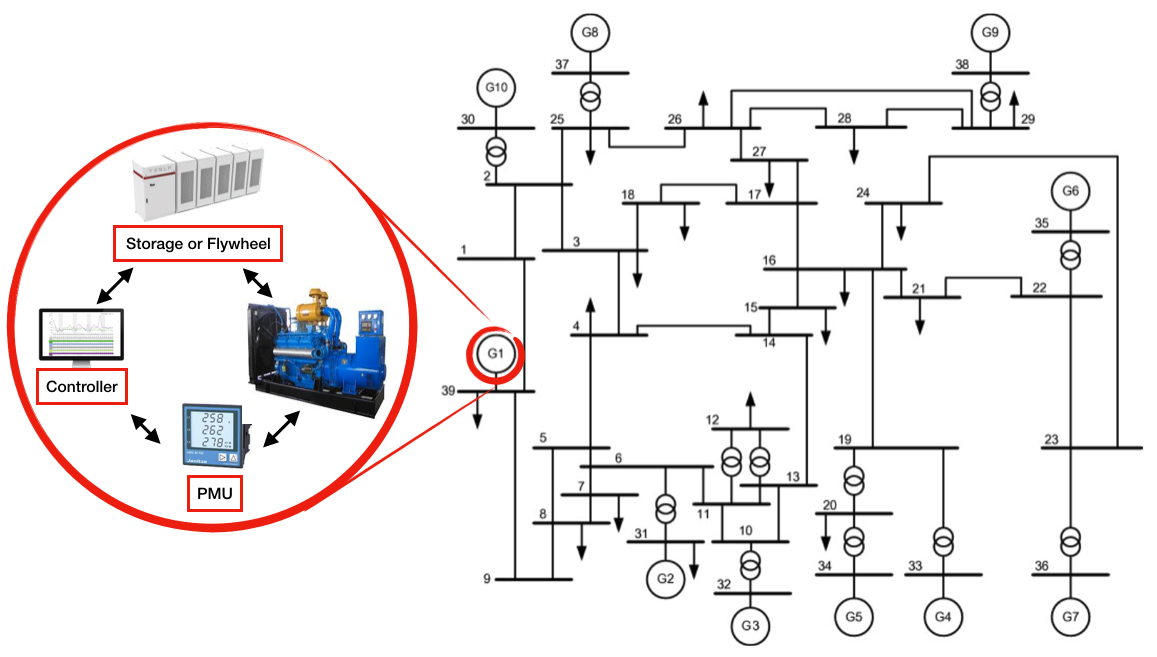} 
	\caption{IEEE 10-generator 39-bus power system network (figure is adapted from \cite{atawi2013advance}). In Example \ref{example_1}, we assume that each generator has both frequency and angle sensors (e.g., PMUs) and also actuator (e.g., storage or flyweel). 	 }
	\label{fig:IEEE-39}
\end{figure}

\section{Numerical Examples}
In this section, we consider a numerical example to demonstrate our results. 
\vspace{.2cm}

\begin{example}[Power Network]\label{example_1}

\begin{figure*}
	\centering
	\begin{subfigure}[b]{0.33\textwidth}
		\centering
	\includegraphics[width=1.01\textwidth]{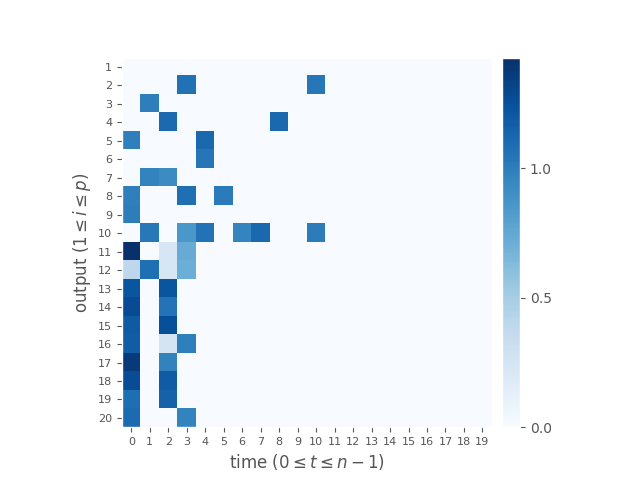} 
	\caption{ }
\end{subfigure}
	\begin{subfigure}[b]{0.33\textwidth}
	\centering
	\includegraphics[width=1.01 \textwidth]{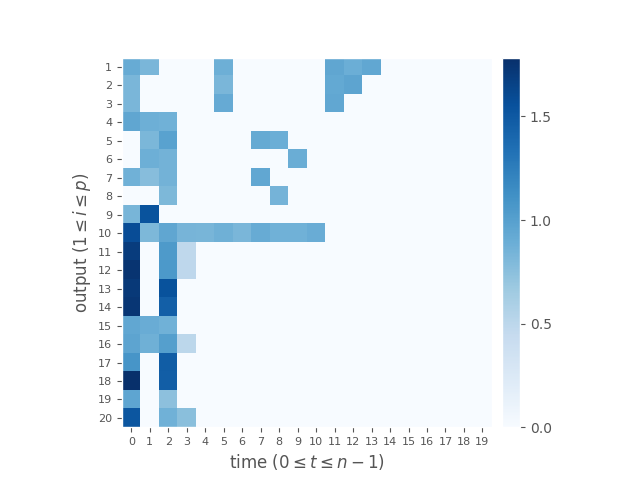} 
	\caption{ }
	\end{subfigure}
	\begin{subfigure}[b]{0.33\textwidth}
	\centering
	\includegraphics[width=1.01 \textwidth]{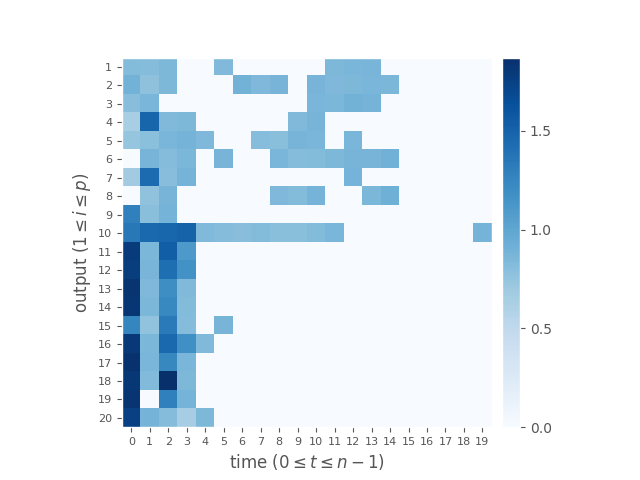} 
	\caption{}  
	\end{subfigure} %
		 \caption{Subplots (a)-(c) present three weighted sparse sensor schedules for Example \ref{example_1} based on the proposed deterministic method (Algorithm \ref{al-approx}) where $d_s \in \{2.2, 3.4, 6.05\}$ is the average number of active sensors, respectively. The color of element $(i,k)$ is proportional to the scaling factor $\mathfrak s^2_i(k)$ where $i \in [20]$ and $k+1 \in [20]$.}\label{fig:1}
\end{figure*}

\begin{figure*}
	\centering
	\begin{subfigure}[b]{0.33\textwidth}
		\centering
	\includegraphics[width=1.01\textwidth]{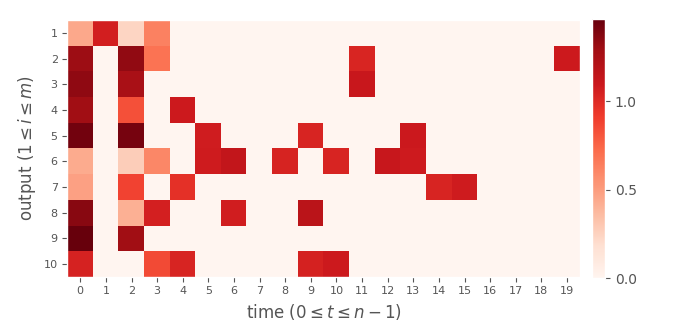} 
	\caption{ }
\end{subfigure}
	\begin{subfigure}[b]{0.33\textwidth}
	\centering
	\includegraphics[width=1.01 \textwidth]{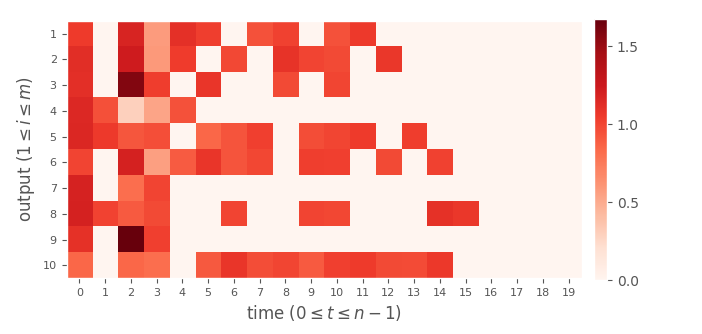} 
	\caption{ }
	\end{subfigure}
		\begin{subfigure}[b]{0.33\textwidth}
	\centering
	\includegraphics[width=.995 \textwidth]{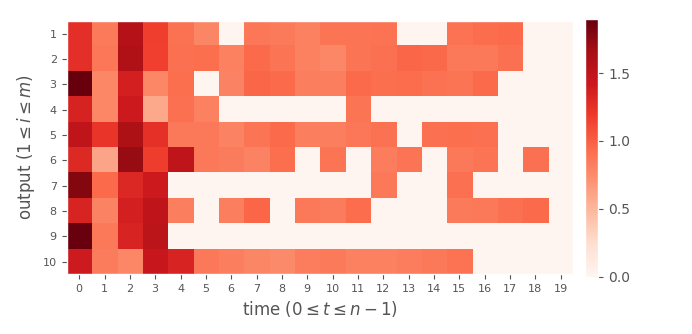} 
	\caption{ }
	\end{subfigure} 
	 \caption{Subplots (a)-(c) present three weighted sparse actuator schedules for Example \ref{example_1} based on the proposed deterministic method (Algorithm \ref{al-approx}) where $d_a \in \{2.3, 3.95, 6.35\}$ is the average number of active actuators, respectively. The color of element $(i,k)$ is proportional to the scaling factor $\mathfrak a^2_i(k)$ where $i \in [20]$ and $k+1 \in [20]$.}\label{fig:2}
\end{figure*}

The problem is to select a set of sensors and actuators to be involved in the wide-area damping control of power systems. We apply our sparse scheduling approach on the IEEE 39-bus test system (a.k.a. the 10-machine New England Power System; see Fig. \ref{fig:IEEE-39}) \citep{atawi2013advance,liu2017minimal}. The single line diagram presented in this figure comprises generators ($G_i$ where $i \in [10]$), loads (arrows), transformers (double circles), buses (bold line segments with number $i \in [39]$), and lines between buses (see \cite{atawi2013advance,liu2017minimal}).

The goal of the wide-area damping control is to damp the fluctuations between generators and synchronize all generators.  The voltage at each generator is adjusted by the control inputs to regulate the power output.

We start with a model representing the interconnection between subsystems. Consider the  continuous-time swing dynamics
	\[ m_i \ddot{\theta_i} + d_i \dot{\theta_i} ~=~- \sum_{j \sim i} k_{ij}(\theta_i - \theta_j)+u_i,\]
where $\theta_i$ is the rotor angle state and $w_i := \dot \theta_i$ is the frequency state of generator $i$.
We assume this power grid model consists of $n=10$ {generators} \citep{liu2017minimal, atawi2013advance}. 
The state space model of the swing equation used for frequency control in power networks can be written as follows 
\begin{eqnarray*}
	\left[ \begin{array}{ccc}
	\dot \theta(t) \\
	\dot w(t) \end{array} \right]  &=& \left[ \begin{array}{ccc}
		0 & I  \\
		-M^{-1}L & -M^{-1} D  \end{array} \right] \left[ \begin{array}{ccc}
		\theta(t) \\
		w(t) \end{array} \right]\nonumber\\
		&&~+~\left[ \begin{array}{ccc}
		0\\
		M^{-1} \end{array} \right] u(t) \label{formation-1}
		\\
 	y(t) &=& \left[ \begin{array}{ccc}
 		\theta(t) \\
 		w(t) \end{array} \right], 	
\end{eqnarray*}
where $M$ and $D$ are diagonal matrices with inertia coefficients and damping coefficients of generators and their diagonals, respectively. 

We assume that both rotor angle and frequency are available for measurement at each generator. This means each subsystem in the power network has a phase measurement unit (PMU). The PMU is a device that measures the electrical waves on an electricity grid using a common time source for synchronization. The system is discretized to the discrete-time LTI system with state matrices $A$, $B$, and $C$ and the sampling time of $0.2$ second (the matrices are borrowed from \cite{ghazal}).

\begin{table*}
\begin{center}
 \caption{The values of approximation factor $\epsilon$ for 20 different  (S/A) schedules based on Algorithm \ref{al-approx} for the power system given in Example \ref{example_1}. 
	      Parameters $d_\text{s}$ and $d_\text{a}$ are defined by \eqref{def-d-s} and \eqref{def-d-a}, respectively.   The last column (in {\color{cyan} cyan} color) shows the value of approximation factor $\epsilon$ for 5 different sensor schedules where all actuators are active (see Definition \ref{def:obs}), and the last row (in {\color{blue} blue} color) shows the value of approximation factor $\epsilon$ for 4 different actuator schedules where all sensors are active (see Definition \ref{def:con}). Moreover, this table validates Theorem \ref{theorem3} as the sum of approximation factors for the two separate sensor and actuator scheduling problems is greater than the approximation factor for the joint S/A schedule and the gap is small. }\label{table_epsi_example}
\begin{tabular}{|c|l|c|c|c|c|}
\cline{3-6}
\multicolumn{2}{c|}{} & \multicolumn{4}{c|}{Average Number of Active Actuators} \\
\cline{3-6}
\multicolumn{2}{c|}{} & $d_{\text a}=2.3$ & $d_{\text a}=3.95$ & $d_{\text a}=6.35$ & {\footnotesize \color{cyan} Fully actuated} \\
\hline
\multirow{4}{*}{\begin{sideways} \footnotesize Avg.\# Active Sensors\, \end{sideways}} 
& $d_{\text s}=2.2$ & ${\bf4.1330} < {\color{blue}1.9434}+{\color{cyan}2.3281}$ & ${\bf  3.4365} <{\color{blue}1.2366}+{\color{cyan}2.3281} $ & ${\bf 2.9078} < {\color{blue}0.6983}+{\color{cyan}2.3281}$ &\color{cyan} $2.3281$  \\ 
\cline{2-6}
& $d_{\text s}=3.4$&${\bf 3.4794}<{\color{blue}1.9434}+{\color{cyan}1.7157}$&$  {\bf 2.7821}<{\color{blue}1.2366}+{\color{cyan}1.7157}$&${\bf 2.2595}<{\color{blue}0.6983}+{\color{cyan}1.7157}$&\color{cyan}$1.7157$\\
\cline{2-6}
& $d_{\text s}=6.05$&${\bf 2.8551}<{\color{blue}1.9434}+{\color{cyan}1.0691}$&${\bf 2.1545}<{\color{blue}1.2366}+{\color{cyan}1.0691}$&${\bf1.6316}<{\color{blue}0.6983}+{\color{cyan}1.0691}$&\color{cyan}$1.0691$\\
\cline{2-6}
& $d_{\text s}=8.85$&${\bf  2.4421}<{\color{blue}1.9434}+{\color{cyan}0.6399}$&${\bf  1.7370}<{\color{blue}1.2366}+{\color{cyan}0.6399}$&${\bf 1.2107}<{\color{blue}0.6983}+{\color{cyan}0.6399}$&\color{cyan}$0.6399$\\
\cline{2-6}
& \footnotesize \color{blue} Fully sensed & \color{blue}$  1.9434$& \color{blue}$1.2366$& \color{blue}$ 0.6983$&$0$\\
\hline
\end{tabular}
\end{center}
\label{Table4}
\end{table*}

Figs. \ref{fig:1} and \ref{fig:2} depict joint sparse sensor and actuator schedules based on the proposed deterministic method (Algorithms \ref{al-approx}) for different values of $d_\text{s}$ and $d_\text{a}$. 

{To have a fair comparison, we normalize the resulting schedules such that the sum of all the scalings satisfies
\[\sum_{k=0}^{n-1} \sum_{i=1}^p \mathfrak s_i^2(k) = nd_{\text{s}},\]
 and 
 \[\sum_{k=0}^{n-1} \sum_{i=1}^m \mathfrak a_i^2(k) = nd_{\text{a}},\]
 where $d_{\text s}$ and $d_{\text a}$ are average numbers of active sensors and actuators, respectively.}

{In Table \ref{table_epsi_example}, the values of approximation factor $\epsilon$ for 20 different  (S/A) schedules based on Algorithm \ref{al-approx} for the power system given in Example \ref{example_1}. The last column (in {\color{cyan} cyan} color) shows the value of approximation factor $\epsilon$ for 5 different sensor schedules where all actuators are active (see Definition \ref{def:obs}) and the last row (in {\color{blue} blue} color) shows the value of approximation factor $\epsilon$ for 4 different actuator schedules where all sensors are active (see Definition \ref{def:con}). Based on Theorem \ref{theorem3} the $(\epsilon, d_{\text s}, d_{\text a})$-joint S/A schedule can be obtained by designing two separate sensor and actuator schedules.
This can be validated by Table \ref{table_epsi_example}, as the sum of approximation factors for the two separate sensor and actuator scheduling problems is greater than the approximation factor for the joint S/A schedule. For example, based on Table \ref{table_epsi_example}, by combining $(1.0691, 6.35)$-actuator schedule and $(0.6983, 6.05)$-actuator schedule, we get joint $(1.6316, 6.35, 6.05)$-S/A schedule. We should note that $1.6316 < 1.0691+0.6983 = 1.7674$; therefore, it follows that this joint schedule is also $(1.7674, 6.35, 6.05)$-S/A schedule (cf. Theorem \ref{theorem3}).

\begin{table*}
\begin{center}
	      \caption{The values of Hankel norm $\sigma_{\max}(\mathcal H_s)$ for 20 different  (S/A) schedules based on Algorithm \ref{al-approx} for the power system given in Example \ref{example_1}. Parameters $d_\text{s}$ and $d_\text{a}$ show average numbers of active sensors and actuators, respectively. For example, $0.1607$ is the value of Hankel norm of the (S/A) schedule with on average $d_{\text s}=2.2$ active sensors and $d_{\text a}=3.95$  active actuators.}\label{table_hankel_example}
\begin{tabular}{|c|l|c|c|c|c|}
\cline{3-6}
\multicolumn{2}{c|}{} & \multicolumn{4}{c|}{Average Number of Active Actuators} \\
\cline{3-6}
\multicolumn{2}{c|}{} & $d_{\text a}=2.3$ & $d_{\text a}=3.95$ & $d_{\text a}=6.35$ & $d_{\text a}=10$ \\
\hline
\multirow{4}{*}{\begin{sideways} \footnotesize Avg.\# Active Sensors\, \end{sideways}} 
& $d_{\text s}=2.2$ &~ $0.1297$~ &~ $0.1607$~ &~ $0.2020$~ &~ $0.2521$  \\
\cline{2-6}
& $d_{\text s}=3.4$&~$0.1672$~&~$  0.2134$~&~$0.2697$~&~0.3429\\
\cline{2-6}
& $d_{\text s}=6.05$&~$0.2271$~&~$0.2935$~&~$0.3718$~&~0.4714\\
\cline{2-6}
& $d_{\text s}=8.85$&~$0.2800$~&~$0.3632$~&~$0.4602$~&~$0.5838$\\
\cline{2-6}
& $d_{\text s}=20$ &~$   0.3790$~&~$0.4981$~&~$ 0.6322$~&~$0.8036$\\
\hline
\end{tabular}
\end{center}
\end{table*}

\begin{table*}
\begin{center}
 \caption{The values of $| \log (\frac{\sigma_{\max}(\mathcal H_s)}{\sigma_{\max}(\mathcal H)})|$ for 20 different  (S/A) schedules based on Algorithm \ref{al-approx} for the power system given in Example \ref{example_1}. The matrix $\mathcal H_s$ and $\mathcal H$ are the Hankel matrices of \eqref{model-aa}-\eqref{model-a} and \eqref{model:bb}-\eqref{model:b}, respectively.
	      Parameters $d_\text{s}$ and $d_\text{a}$ are defined by \eqref{def-d-s} and \eqref{def-d-a}, respectively.  For example, $| \log (\frac{\sigma_{\max}(\mathcal H_s)}{\sigma_{\max}(\mathcal H)})|=0.7706$ for the (S/A) schedule with on average $d_{\text s} =6.05$ active sensors and $d_{\text a}=6.35$ active actuators. }\label{table_error_example}
\begin{tabular}{|c|l|c|c|c|c|}
\cline{3-6}
\multicolumn{2}{c|}{} & \multicolumn{4}{c|}{Average Number of Active Actuators} \\
\cline{3-6}
\multicolumn{2}{c|}{} & $d_{\text a}=2.3$ & $d_{\text a}=3.95$ & $d_{\text a}=6.35$ & $d_{\text a}=10$ \\
\hline
\multirow{4}{*}{\begin{sideways} \footnotesize Avg.\# Active Sensors\, \end{sideways}} 
& $d_{\text s}=2.2$ & $1.8237$ & $1.6094$ & $1.3808$ & $1.1592$  \\ 
\cline{2-6}
& $d_{\text s}=3.4$&~$1.5698$~&~$  1.3260$~&~$1.0918$~&~$0.8515$\\
\cline{2-6}
& $d_{\text s}=6.05$&~$1.2634$~&~$1.0070$~&~$0.7706$~&~$0.5333$\\
\cline{2-6}
& $d_{\text s}=8.85$&~$1.0543$~&~$0.7940$~&~$0.5573$~&~$0.3194$\\
\cline{2-6}
& $d_{\text s}=20$ &~$  0.7515$~&~$0.4783$~&~$ 0.2399$~&~$0$\\
\hline
\end{tabular}
\end{center}
\end{table*}

In Tables \ref{table_hankel_example} and \ref{table_error_example}, the values of Hankel norm $\sigma_{\max}(\mathcal H_s)$ and relative error $| \log (\frac{\sigma_{\max}(\mathcal H_s)}{\sigma_{\max}(\mathcal H)})|$ for 20 different  (S/A) schedules are calculated based on Algorithm \ref{al-approx}, respectively.

The sparsity degree of each schedule is captured by  $d_\text{s}$ and $d_\text{a}$. Based on Table \ref{table_hankel_example}, as parameter $d_\text{a}$ (i.e., the average number of active actuators) increases both the number of non-zero scalings (i.e., activations) and the Hankel norm increase.
Similarly, as parameter $d_\text{s}$ (i.e., the average number of active sensors) increases both the number of non-zero scalings (i.e., activations) and the Hankel norm increase. As one expects, according to Table \ref{table_error_example}, the value of $| \log (\frac{\sigma_{\max}(\mathcal H_s)}{\sigma_{\max}(\mathcal H)})|$ decreases as the number of active sensors or actuators increases.

Finally, we should note that the value of $| \log (\frac{\sigma_{\max}(\mathcal H_s)}{\sigma_{\max}(\mathcal H)})|$ is less than the value of corresponding approximation factor $\epsilon$ (cf. Tables \ref{table_epsi_example} and \ref{table_error_example}). The values of $| \log (\frac{\rho(\mathcal H_s)}{\rho(\mathcal H)})|$ in fact  is less than $\epsilon$ for any Hankel-based performance measures based on Definition \ref{def:S/A} and \eqref{518}. }

\end{example}

\section{Concluding Remarks} 
In this paper, we studied the problem of designing a joint sparse sensor and actuator (S/A) schedule of linear dynamical systems that retains the full observability and controllability of the system.
Based on recent advances in matrix reconstruction and graph sparsification literature, we provide a polynomial-time joint S/A schedule for a discrete time linear dynamical system. This joint S/A schedule on average selects only a constant number of sensors and actuators at each time step, while guaranteeing a control/estimation performance that approximates the fully sensed/actuated setting. We further prove the validity of separation principle for the system, showing that the problem can be decomposed into finding sensor and actuator schedules separately.

\bibliographystyle{IFAC}        
\begin{spacing}{1.1}
\bibliography{main_Milad}
\end{spacing}
\appendix

\addtolength{\textheight}{-12cm}

\end{document}